\documentclass[10pt]{article}
\usepackage{amsmath,amssymb,amsthm,mathptmx,mathrsfs,cite}
\numberwithin{equation}{section}
\newtheorem{theorem}{Theorem}[section]
\newtheorem{corollary}[theorem]{Corollary}
\newtheorem{lemma}[theorem]{Lemma}
\newtheorem{proposition}[theorem]{Proposition}
\theoremstyle{definition}

\newtheorem{remark}[theorem]{Remark}

\DeclareMathOperator{\sgn}{sgn}
\def \f {\frac}
\def \k {\kappa}

\begin{document}

\title{Convergence of the dispersion Camassa-Holm $N$-Soliton}

\author{Fengfeng Dong\thanks{School of Mathematical Sciences, Tongji University, Shanghai 200092, P.R. China; 1433494@tongji.edu.cn}\and Lingjun Zhou\thanks{School of Mathematical Sciences, Tongji University, Shanghai 200092, P.R. China; zhoulj@tongji.edu.cn}}
\date{\today}
\maketitle

\begin{abstract}
In this paper, we show that the peakon (peaked soliton) solutions can be recovered from the smooth soliton solutions, in the sense that there exists a sequence of smooth $N$-soliton solutions of the dispersion Camassa-Holm equation converging to the $N$-peakon of the dispersionless Camassa-Holm equation uniformly with respect to the spatial variable $x$ when the dispersion parameter tends to zero. The main tools are asymptotic analysis and determinant identities.

\linespread{2.0}\selectfont

\noindent{\textbf{Keywords}}  Camassa-Holm equation, peakon, soliton, asymptotic analysis, determinant technique

\noindent{\textbf{MSC(2010)}} 35Q51, 35C08, 15A15
\end{abstract}

\section{Introduction}
Consider the Camassa-Holm (CH) equation \cite{CH1993}
\begin{equation}\label{kCH}
u_t+2\omega u_x-u_{xxt}+3uu_x=2u_x u_{xx}+u u_{xxx},
\end{equation}
where $u$ is the fluid velocity in the $x$ direction, $\omega$ is a constant related to the critical shallow water speed, and the subscripts denote partial derivatives. The CH equation models the unidirectional propagation of shallow water waves over a flat bottom \cite{CH1993, CHH1994, Johnson2002, CL2009}.

When the dispersion coefficient $\omega>0$, we refer to \eqref{kCH}  as the dispersion CH equation in this paper, which possesses solitary wave solutions whose limiting form as $\omega\rightarrow 0$ has peaks where the first derivative are discontinuous \cite{CH1993, CHH1994}.  There are abundant literature on the soliton solutions, involving serval techniques in the soliton theory; see \cite{ C2000, C1998isp, RSJ, LZ, LYS, MY1, P1, P2, P3, P4, XZQ} for details.

When  $\omega=0$, the equation \eqref{kCH}, i.e.
\begin{equation}\label{CH eqn}
u_t-u_{xxt}+3uu_x=2u_x u_{xx}+u u_{xxx},
\end{equation}
compared with the KdV equation, has the advantage of admitting peakons(peaked solitons) of the form $ce^{-|x-ct|}$, which capture the main feature of the exact traveling wave solutions of greatest height of the governing equations for water waves in irrotational flow \cite{Constantin2006, Constantin2007, Constantin2012}, and modeling wave breaking \cite{CHH1994, BSS2000, C1998wavebreaking}. Equation \eqref{CH eqn} admits multi-soliton solution of the form \cite{CH1993}
\begin{equation}\label{N-peakon}
    u(x,t)=\sum_{i=1}^Nm_i(t)e^{-|x-x_i(t)|},
\end{equation}
which is called multipeakons. The system of evolution equations for amplitudes $m_i(t)$ and positions $x_i(t)$ of peaks is a completely integrable finite dimensional Hamiltonian system, explicit formulas of which have been obtained by Beals \textsl{e.l.} \cite{BSS1999, BSS2000} via the inverse scattering transform method using the Stieltjes theorem on continued fractions. Blow-up results for certain initial data and global existence theorems for for a large class of initial data have been given by \cite{CH1993, CHH1994} and \cite{C1998blowup}, respectively.

Note that \eqref{CH eqn} is a formal limit of \eqref{kCH}, a natural question is to analyze the relation between the parameter solitons of \eqref{kCH} and the peakons of \eqref{CH eqn}.
Some theoretical and numerical results have been obtained on this question. Li and Olver \cite{LiO1, LiO2} have showed how the single peakon of \eqref{CH eqn} can be recovered as limits of classical solitary wave solutions via the theory of dynamical systems. Johnson \cite{RSJ} and Parker \cite{P1} have claimed that the profile of the analytic one, two-soliton solutions may become peaked for some singular limit involving the wave parameters and dispersion parameter in their settings. These answers obtained are not so satisfactory.

In this paper, we prove rigorously that the non-analytic $N$-peakon (amplitudes $m_i$ in \eqref{N-peakon} are all positive) solutions can be recovered from the analytic $N$-soliton  solutions given by \cite{LZ, LYS}. Our work is motivated by \cite{LiO1, LiO2, MY2}, and establishes the relation between solutions of \eqref{kCH} and \eqref{CH eqn}. We believe that there is certain physics behind this mathematical fact, while we still cannot explain it completely.

The layout of this paper is as follows. In Section 2, we simply review some results on explicit solutions of  \eqref{kCH} and \eqref{CH eqn}, and prove some determinant identities that are important to the proof of convergence.
In Section 3, results on the convergence of the $N$-soliton are proposed.  In Section 4, we prove the  convergence with the aid of asymptotic analysis and determinant technique.

\section{Explicit solutions to Camassa-Holm equation and determinant identities}

 Let $\omega=\k^2$, and without lose of generality, we assume that $\k>0$. The $N$-soliton solutions of \eqref{kCH} given by Li and Zhang \cite{LZ, LYS} reads:
\begin{equation*}
  u(y,t)= \left(\ln\,\Big|\f{f_1}{f_2}\Big|\right)_t, \qquad  x(y,t)= \ln\,\Big|\f{f_1}{f_2}\Big|,
\end{equation*}
where
\begin{equation*}
f_1=\f{W(\Phi_1,\ldots,\Phi_N,e^{\f{y}{2\k}})}{W(\Phi_1,\ldots,\Phi_N)},\qquad
f_2=\f{W(\Phi_1,\ldots,\Phi_N,e^{-\f{y}{2\k}})}{W(\Phi_1,\ldots,\Phi_N)},
\end{equation*}
\[\Phi_i(y,t)=
\begin{cases}
\cosh\xi_i,\qquad i=2l-1,\\
\sinh\xi_i,\qquad i=2l,
\end{cases}
\qquad 1\leq i\leq N,\]
and
\[\xi_i=k_i(y-\k c_it),\quad k_i=\f{1}{2\k}\left(1-\f{2\k^2}{c_i}\right)^{\f{1}{2}},\quad k_1<\cdots<k_N.\]

\begin{remark}
  According to the examples in \cite {LYS1} and our proof in section \ref{proof2.1}, we suppose that $0<2\k k_i<1\,(1\leq i\leq N)$, then the formulas above really give smooth $N$-soliton and  \[c_i= \f{2\k^2}{1-4\k^2k_i^2}>2\k^2\,(1\leq i\leq N).\]
\end{remark}

Let $I=\{i_1<i_2\cdots <i_n\}$ be an n-element subset of the integer interval $[1,N]=\{1,2,\cdots,N\}$, and $J=\{1,2,\cdots,N\}\backslash I$.
Define $a_i$ and $b_i$ as follows:
\[a_i=\f{1+2\k k_i}{1-2\k k_i},\qquad b_i=\f{1-2\k k_i}{1+2\k k_i},\]
and set
\[a_I=\prod_{i\in I}a_i,\qquad b_I=\prod_{i\in I}b_i,\qquad \xi_I=\sum_{i\in I}\xi_i,\qquad\Gamma_I=\prod_{i,j\in I,\,i<j}(k_j-k_i),\]
with the proviso that $\Gamma_{\{i\}}=1$.

We now state some results on $N$-soliton solutions of \eqref{kCH}.

\begin{theorem}\label{new expression}
The $N$-soliton given by Li and Zhang can be expressed as follows,
\begin{equation}\label{u(y,t)}
  u(y,t) = \left(\ln\f{g_1}{g_2}\right)_t,
\end{equation}
\begin{equation}\label{x(y,t)}
     x(y,t)=\frac{y}{\k}+\ln\frac{g_1}{g_2}+\alpha,
\end{equation}
where
\begin{subequations}
\begin{equation}\label{g_1}
 g_1=1+\sum_{n=1}^N\left(\sum_Ib_Ie^{2\xi_I}\prod_{  j\in J,\,l\in I}
   \sgn (j-l)\frac{k_j+k_l}{k_j-k_l}\right),
  \end{equation}
\begin{equation}\label{g_2}
 g_2=1+\sum_{n=1}^N\left(\sum_Ia_Ie^{2\xi_I}\prod_{j\in J,\,l\in I}
   \sgn(j-l)\frac{k_j+k_l}{k_j-k_l}\right).
  \end{equation}
\end{subequations}
\end{theorem}

 We have left the proof of Theorem \ref{new expression} to section \ref{proof2.1}, the one who is not interested in can skip it without affecting the comprehension of the main results.

\begin{theorem}\label{equivalence}
  For  the smooth $N$-soliton of \eqref{kCH}, the result of Li and Zhang obtained via Darboux transformation is equivalent to Parker's proposed by Hirota bilinear method.
\end{theorem}

\begin{proof}
  Let
\begin{equation}\label{transformation}
    \varphi_i=\ln\left( \prod_{j=1, \,j\neq i}^N \sgn (j-i)\f{k_j+k_i}{k_j-k_i}\right),
\end{equation}
and $\sum_{\mu=0,1}$ be the summation over all possible combination of $\mu_1=0,1,\,\ldots ,\mu_N=0,1$. Under the transformation: $\xi_i\rightarrow \xi_i-\f{\varphi_i}{2},\;k_i\rightarrow\f{k_i}{2}$, the $N$-soliton solutions of \eqref{kCH} with arbitrary initial phase becomes
\begin{equation}\label{parameter solution}
  u(y,t;\k) = \left(\ln\f{g_1}{g_2}\right)_t, \qquad
  x(y,t;\k) =\f{y}{\k}+ \ln\f{g_1}{g_2}+\alpha,
\end{equation}
where
 \begin{equation}\label{g_1g_2}
    \begin{aligned}
    g_1 =\sum_{\mu=0,1}\exp\left(\sum_{i=1}^N\mu_i(\xi_i-\phi_i)+\sum_{i<j}^N \mu_i\mu_j\gamma_{ij}\right), \\
 g_2 = \sum_{\mu=0,1}\exp\left(\sum_{i=1}^N\mu_i(\xi_i+\phi_i)+\sum_{i<j}^N \mu_i\mu_j\gamma_{ij}\right),
\end{aligned}
\end{equation}
with
\begin{equation}\label{c_i}
 \xi_i=k_i(y-\k c_it-y_{i0}), \quad  c_i = \frac{2\k^2}{1-\k^2k_i^2},\quad 0<\k k_i<1\quad (1\leq i\leq N)
\end{equation}
\[\phi_i=\ln \f{1+\k k_i}{1-\k k_i}, \quad e^{\gamma_{ij}}=\left(\frac{k_i-k_j}{k_i+k_j}\right)^2,\quad 0<k_1<\cdots<k_N,\]
and $\alpha,\,y_{i0},\,(1\leq i\leq N)$ are arbitrary constants.
\end{proof}

We now review the $N$-peakon of the dispersionless CH equation \eqref{CH eqn}.

Consider the spectral problem associated to \eqref{CH eqn}:
\begin{equation}\label{spectral problem}
\psi_{xx}=(\f{1}{4}+\lambda m)\psi, \qquad e^{\pm\f{x}{2}}\psi(x)\rightarrow 0\quad \text{as}\quad x\rightarrow \mp\infty
\end{equation}
with $m=2\sum_{i=1}^Nm_i\delta_{x_i}$. Denote the eigenvalues of \eqref{spectral problem} by $\tilde{\lambda}_i$\,$(i=1,\ldots, N)$, let $\Delta_n^m$  be the determinant of the $n\times n$ submatrix of a infinite Hankel matrix, whose $(1,1)$ entry is $\tilde{A}_m$, i.e. $\Delta_n^m=\det (\tilde{A}_{m+i+j})_{i,j=0}^{n-1}$, the moments $\tilde{A}_m$ are restricted by
\[\tilde{A}_m=\sum_{i=0}^N(-\tilde{\lambda}_i)^m a_i, \quad a_i=a_i(0)e^{-\f{t}{2\tilde{\lambda}_i}}\,(i\geq 0),\quad \tilde{\lambda}_0=0,\quad a_0=\f{1}{2}.\]
Let $\tilde{\Delta}_n^m=\det (\hat{A}_{m+i+j})_{i,j=0}^{n-1}$ with
\[\hat{A}_m=\sum_{i=1}^N(-\tilde{\lambda}_i)^m a_i, \quad a_i=a_i(0)e^{-\f{t}{2\tilde{\lambda}_i}}.\]

\begin{lemma}[$N$-peakon, \cite{BSS2000}]\label{lem:N-peakon}
The dispersionless CH equation \eqref{CH eqn} admits $N$-peakon solutions
\begin{equation*}
u(x,t)=\sum_{i=1}^Nm_i(t)e^{-|x-x_i(t)|},
\end{equation*}
where
\begin{equation}\label{m_i}
   m_i= \f{2\tilde{\Delta}_{N-i+1}^0\Delta_{N-i}^2}{\Delta_{N-i+1}^1\Delta_{N-i}^1},
   \qquad  x_i = \ln\left(\f{2\tilde{\Delta}_{N-i+1}^0}{\Delta_{N-i}^2}\right).
\end{equation}
\end{lemma}

Last, we present some determinant identities necessary to the proofs of the convergence of smooth $N$-soliton.

Let
 \begin{equation}\label{A_m}
    A_m=\sum_{i=1}^N\lambda_i^mE_i,\quad E_i=e^{\f{2}{\lambda_i}t+x_{i0}},\quad x_{i0}=\f{y_{i0}}{\k},\quad\lambda_i= \f{2}{c_i},
\end{equation}
 where $m,n\in \mathbb{Z}$ and $0\leq n\leq N$. Let $D_n^m=\det (A_{m+i+j})_{i,j=0}^{n-1}$, and adopt the convention that $D_0^m=1$.

\begin{lemma}
For $ n=1,\ldots, N$,
  \begin{equation}\label{Hankle determinant}
    D_n^m=\sum_{1\leq i_1<\cdots <i_n\leq N}\Delta_n(i_1,\cdots , i_n)(\lambda_{i_1}\cdots \lambda_{i_n})^mE_{i_1}\cdots E_{i_n},
  \end{equation}
where
  \begin{equation*}
    \Delta_n(i_1,\cdots , i_n)=\prod_{1\leq l<m\leq n}(\lambda_{i_l}-\lambda_{i_m})^2.
  \end{equation*}
\end{lemma}

 We adopt the convention that $D_0^m=1$ and $D_n^m=0\,(n<0)$. Since $\lambda_i's$ are distinct, we have $D_n^m>0\,(0\leq n\leq N)$, and  $D_n^m=0\,(n>N)$ by \cite[Theorem 6.1]{BSS2000}.

\begin{lemma}[Jacobi identity, \cite{VD}]\label{Jacobi identity}
For any determinant D, let $D(i,j;k,l)$ be  the determinant obtained from D by deleting the rows $i,j$ and the columns $k,l$, respectively, then
 \begin{equation*}
         DD(i,j;k,l) =D(i;k) D(j;l) -  D(j;k)D(i;l),  \qquad i<j,\quad k<l,
 \end{equation*}
\end{lemma}

\begin{lemma}\label{critical identities}
For any $m,n\in\mathbb{Z}$ and $n\geq0$,
  \begin{equation}\label{Jacobi formula}
  D_{n+2}^mD_n^{m+2}=D_{n+1}^{m+2}D_{n+1}^m-(D_{n+1}^{m+1})^2,
  \end{equation}
\begin{subequations}
   \begin{equation}\label{Jacobi formula 1}
    D_{n+1,t}^2D_n^2-D_{n,t}^2D_{n+1}^2=2D_{n+1}^1D_n^3,
  \end{equation}
  \begin{equation}\label{Jacobi formula 2}
    D_{n+2,t}^0D_{n+1}^0-D_{n+1,t}^0D_{n+2}^0=2D_{n+2}^{-1}D_{n+1}^1,
  \end{equation}
\end{subequations}
  \begin{equation}\label{Jacobi formula 3}
    D_{n,t}^2D_{n+2}^0-D_{n+1,t}^2D_{n+1}^0=-2D_{n+2}^{-1}(1;2)D_{n+1}^1,
  \end{equation}
  \begin{equation}\label{Jacobi formula 4}
    D_{n+1,t}^0D_{n+1}^2-D_{n+2,t}^0D_n^2=2D_{n+2}^{-1}(1;2)D_{n+1}^1.
  \end{equation}
\end{lemma}

\begin{corollary}\label{identities for main theorem}
\begin{equation}\label{formula}
    \sum_{i=0}^n\f{(D_i^2)^2}{D_{i+1}^1D_i^1}=\f{D_n^3}{D_{n+1}^1},
  \qquad   \sum_{i=n+1}^{N-1}\f{(D_{i+1}^0)^2}{D_{i+1}^1D_i^1}=\f{D_{n+2}^{-1}}{D_{n+1}^1}.
\end{equation}
\end{corollary}

\begin{proof}[\bf{Proof of Lemma \ref{critical identities}}]
The identity  \eqref{Jacobi formula} follows from Lemma \ref{Jacobi identity} with $D=D_n^m$ and $i=k=1,j=l=n+2$.
By the characteristics of  the determinant $D_{n}^m$, \eqref{Jacobi formula 2} is equivalent to \eqref{Jacobi formula 1}.
 Note that $A_m$ is given by \eqref{A_m}, we have $A_{m,t}=2A_{m-1}$, therefore,
  \[m 2D_{n+1}^{m-2}(1;2)=D_{n,t}^m= 2D_{n+1}^{m-1}(n+1;2).\]
Using the relation above, \eqref{Jacobi formula 1} and \eqref{Jacobi formula 3}-\eqref{Jacobi formula 4} can be rewritten as follows,
\begin{align*}
  D_{n+2}^1(n+2;2)D_n^2
  -D_{n+1}^2D_{n+1}^1(n+1;2)&=D_{n+1}^1D_n^3,\\
 D_{n+1}^0(1;2)D_{n+2}^0-D_{n+2}^0(1;2)D_{n+1}^0&=-D_{n+2}^{-1}(1;2)D_{n+1}^1,\\
D_{n+2}^{-1}(n+2;2) D_{n+1}^2-D_{n+3}^{-1}(n+3;2)D_{n}^2&=D_{n+2}^{-1}(1;2)D_{n+1}^1.
\end{align*}
Consider the following determinants of order $n+2$,
\begin{eqnarray*}
   D_1=
    \begin{vmatrix}
     0& A_1 & A_2& \cdots & A_{n+1}\\
     1&  A_2 & A_3& \cdots & A_{n+2}\\
     0&  A_3 & A_4& \cdots & A_{n+3}\\
     \vdots & \vdots & \vdots & \ddots &  \vdots \\
     0& A_{n+2} & A_{n+3}& \cdots & A_{2n+2}\\
    \end{vmatrix},  \quad
    D_3=D_{n+2}^0,\quad
     D_4=
    \begin{vmatrix}
      A_{-1} & A_1& \cdots & A_{n+1}\\
      A_0 & A_2& \cdots & A_{n+2}\\
      \vdots & \vdots & \ddots & \vdots \\
      A_n & A_{n+2}& \cdots & A_{2n+2}\\
    \end{vmatrix}.
\end{eqnarray*}
The identities \eqref{Jacobi formula 1} and \eqref{Jacobi formula 3}-\eqref{Jacobi formula 4} follow readily from Lemma \ref{Jacobi identity}:
\begin{itemize}
  \item
For $D_1$ and $D_4$, set $i=k=1,j=l=n+2$, we have \eqref{Jacobi formula 1} and \eqref{Jacobi formula 4},
  \item
For $D_3$, set $i=1,k=2,j=l=n+2$, we have \eqref{Jacobi formula 3}.
\end{itemize}
\end{proof}

By \eqref{Jacobi formula}, we have
 \[(D_i^2)^2=D_i^1D_i^3-D_{i+1}^1D_{i-1}^3,
  \qquad (D_{i+1}^0)^2=D_{i+1}^{-1}D_{i+1}^1-D_{i+2}^{-1}D_i^1.\]
Dividing the equations above by $D_{i+1}^1D_i^1$ and summing up, we obtain Corollary \ref{identities for main theorem} by taking account of the fact that
$D_n^m=0$ for any $n>N$ or $n<0$.

\section{Results on convergence}

 Motivated by Li and Olver \cite{LiO1, LiO2}, we consider a sequence of $N$-soliton solutions with the velocity of the $i$-th soliton in $(x,t)$-space independent of the dispersion parameter $\k$.
\begin{remark}\label{ss}
There exists a sequence of analytic $N$-soliton solutions, the velocity of the $i$-th soliton in $(x,t)$-space given by positive constant independent of $\k$. In fact,
since $k_i\,'s$ are arbitrary and satisfy $0<\k k_i<1\,(1\leq i\leq N)$, $0<k_1<\cdots<k_N$, let $\alpha_i\,(1\leq i\leq N)$ be N distinct positive constants independent of $\k$, and
\begin{equation}\label{k_i}
  k_i=\f{1}{\k}\left(1-\f{2\k^2}{\alpha_i}\right)^{\f{1}{2}},
\end{equation}
then $\tilde{c}_i$ given by \eqref{c_i} with $k_i$ chosen by \eqref{k_i} equals to $\alpha_i$. Thus the sequence of $N$-soliton solutions given by \eqref{parameter solution}-\eqref{c_i} and \eqref{k_i} is required.
\end{remark}
In the following, we will still use $c_i$ to denote the constant $\tilde{c}_i$, and give the results on the convergence of this sequence of smooth $N$-soliton solutions.
With the preparations above, we have the convergence of the $N$-soliton solutions of the dispersion CH equation.

\begin{theorem}\label{parameter limit}
Let
\begin{equation}\label{x_i}
      \bar{x}_n = \ln\left(\f{2D_{N-n+1}^0}{D_{N-n}^2}\right), \quad 1\leq n\leq N.
\end{equation}
Under appropriate  translation and  scaling transformations, for any given $t$, when $\k\rightarrow 0$, the $N$-soliton given by \eqref{parameter solution}-\eqref{c_i} and \eqref{k_i} has the limit
\begin{equation}\label{n-peakon}
\bar{u}(x;t)= \begin{cases}
 u_n,\quad  \bar{x}_{n-1}<x\leq \bar{x}_n\,(n=1,2,\ldots,N)\\
  u_{N+1}, \quad x>\bar{x}_N
\end{cases}
\end{equation}
with
\begin{equation}\label{limit}
u_n=\f{e^xD_{N-n}^3+4e^{-x}D_{N-n+2}^{-1}}{D_{N-n+1}^1}, \qquad
u_{N+1}=2\sum_{i=1}^Nc_iz_i^{-1}.
\end{equation}
\end{theorem}

\begin{remark}\label{x_n:explaination}
If we set $\lambda_i=-4\tilde{\lambda}_i$, i.e. $c_i=-\f{1}{2\tilde{\lambda}_i}$, which is the asymptotic velocity of the position $x_i$ at large positive time for the peakon solutions given by Lemma \ref{N-peakon}, and $c_i's$ are distinct positive constnts by \cite[Theorems 4.1 and 6.4]{BSS2000}. Then $A_m=4^m\hat{A}_m$ and $D_n^m=4^{n(m+n-1)}\tilde{\Delta}_n^m$ for any $m,n\in \mathbb{Z}$ and $0\leq n\leq N$, hence,
\[D_{N-i+1}^0=4^{(N-i+1)(N-i)}\tilde{\Delta}_{N-i+1}^0,\quad \quad D_{N-i}^2=4^{(N-i)(N-i+1)}\tilde{\Delta}_{N-i}^2.\]
Since $\Delta_n^m=\tilde{\Delta}_n^m( m\geq 1)$,  we obtained that the second expression in \eqref{m_i} is fixed under the transformation: \[\Delta_n^2\rightarrow D_n^2,\quad \tilde{\Delta}_n^0\rightarrow D_n^0.\]
Therefore, when $\lambda_i=-4\tilde{\lambda}_i$, we have $\bar{x}_n=x_n\,(1\leq n\leq N)$.
\end{remark}

\begin{theorem}\label{main result}
 Under appropriate  translation and  scaling transformations, the $N$-soliton solutions of \eqref{kCH} given by \eqref{parameter solution}-\eqref{c_i} and \eqref{k_i} with $\alpha_i=-\f{1}{2\tilde{\lambda}_i}$ converges to the $N$-peakon of \eqref{CH eqn} uniformly with respect to $x$ as $\k\rightarrow 0$.
\end{theorem}

\section{Proof of convergence}
 This section is devote to the proof of Theorems \ref{parameter limit} and \ref{main result}, the asymptotic analysis is the main tool.

Let $g=g_1/g_2$, then under the inverse of the reciprocal transformation, i.e.\,$dx=\frac{1}{r}dy+udt$, we have
 $u(x(y,t),t;\k)=u(y,t;\k)=(\f{\partial}{\partial t}+u\f{\partial}{\partial x})\ln g$, and
 \begin{equation}\label{u(x,t;k)}
   u(x,t;\k)=\f{g_t}{g-g_x}.
\end{equation}
Therefore, for Theorem \ref{parameter limit}, we only need to show: for any given $t$, under appropriate translation and scaling transformations, the limit on $\k\rightarrow 0$ of the right hand side of \eqref{u(x,t;k)} exists and is given by \eqref{n-peakon}. We present some lemmas and propositions to show the existence of this limit, then Theorems \ref{parameter limit} and \ref{main result} follow from  some determinant identities listed in Section 2.

\begin{lemma}\label{asymoptotic}
Let $z_i=\exp(x-c_it-x_{i0})=e^xE_i^{-1}$, $\epsilon=\f{\k^2}{4}$, when
$\epsilon \rightarrow 0$,
\begin{align*}
  &  \k k_i= 1-\f{\k^2}{c_i}+o(\epsilon),\qquad v_i=e^{-\phi_i}=
   \lambda_i\epsilon+o(\epsilon),\\
   &  e^{\xi_i} = \f{z_i}{g}+o(1), \qquad                                  \gamma_{ij}=\epsilon^2(\lambda_i-\lambda_j)^2+o(\epsilon^2).
\end{align*}
\end{lemma}

Under the phase shift: $\xi_i\rightarrow \xi_i-\phi_i$, when $\k\rightarrow 0$ i.e. $\epsilon \rightarrow 0$, we have
\begin{align*}
  & g_1= 1+\sum_{n=1}^N\epsilon^{n(n+1)}g^{-n}\sum_{1\leq i_1<\cdots <i_n\leq N}(\lambda_{i_1}\cdots \lambda_{i_n})^2\Delta_n(i_1,\cdots , i_n)z_{i_1}\cdots z_{i_n}+o(\epsilon^{N(N+1)}),\\
  & g_2= 1+\sum_{n=1}^N\epsilon^{n(n-1)}g^{-n}\sum_{1\leq i_1<\cdots <i_n\leq N}\Delta_n(i_1,\cdots , i_n)z_{i_1}\cdots z_{i_n}+o(\epsilon^{N(N-1)}).
\end{align*}

Replace $z_i\,(i=1,\ldots, N )$ by
\[\f{\prod_{j=1}^N\lambda_j^2}{2\prod_{j=1,j\neq i}^N(\lambda_i-\lambda_j)^2}\f{z_i}{\lambda_i^2},\]
i.e. under scaling transformations, we obtain the following estimate involving $D_n^m$ for $g_1$, $g_2$.

 \begin{lemma}\label{g_1,g_2}
 When $\k\rightarrow 0$ i.e. $\epsilon \rightarrow 0$,
  \begin{align*}
    &g_1 = 1+2\sum_{n=1}^N\epsilon^{n(n+1)}g^{-n}d_{n+1}e^{nx}D_{N-n}^0+o(\epsilon^{N(N+1)}),\\
     &g_2= 1+\sum_{n=1}^N\epsilon^{n(n-1)}g^{-n}d_ne^{nx}D_{N-n}^2+o(\epsilon^{N(N-1)}),
  \end{align*}
with $d_n>0$ defined by
  \begin{equation*}
    d_n(t)=\frac{\prod_{i=1}^N\lambda_i^{2(n-1)}}{2^n\Delta_N\prod_{i=1}^NE_i},\quad (n=1,2,\ldots,N+1).
  \end{equation*}
\end{lemma}

\begin{proof}
  For any $1\leq n\leq N$, let $I_n=\{i_1<i_2<\ldots<i_n\}$ and $J_n=\{1,\ldots,N\}\backslash I_n$ be two ordered subsets of the integer interval $[1,N]=\{1,2,\cdots,N\}$. Denote by $Q$ the following product
  \[\prod_{j=1,\,j\neq i_1}^N(\lambda_{i_1}-\lambda_j)^2\cdots\prod_{j=1,\,j\neq i_n}^N(\lambda_{i_n}-\lambda_j)^2,\]
 then
 \begin{equation*}
    Q=\Delta_n^2\prod_{l\in I_n,\,m\in J_n}(\lambda_l-\lambda_m)^2,
\end{equation*}
with the convention that $\Delta_{\varnothing}=\Delta_{\{i\}}=1$ and $\prod_{l\in I_n,\,m\in J_n}(\lambda_l-\lambda_m)^2=1$ for $J_n={\varnothing}$.
Note that \begin{equation*}
    \Delta_N=\Delta_n\Delta_{N-n}\prod_{l\in I_n,\,m\in J_n}(\lambda_l-\lambda_m)^2,
\end{equation*}
thus we have $ Q=\f{\Delta_n\Delta_N}{\Delta_{N-n}}$. Under the phase shift
\[z_i\rightarrow \f{\prod_{j=1}^N\lambda_j^2}{2\prod_{j=1,\,j\neq i}^N(\lambda_i-\lambda_j)^2}\f{z_i}{\lambda_i^2},\quad i \in I_n,\]
\begin{equation*}
       \begin{split}
         z_{i_1}\cdots z_{i_n}
     =&\f{\Delta_n\prod_{j=1}^N\lambda_j^{2n}e^{nx}E_{i_1}^{-1}\cdots E_{i_n}^{-1}}
     {2^n(\lambda_{i_1}\cdots \lambda_{i_n})^2\prod_{j=1,j\neq i_1}^N(\lambda_{i_1}-\lambda_j)^2\cdots\prod_{j=1,j\neq i_n}^N(\lambda_{i_n}-\lambda_j)^2}\\
     =& \f{\prod_{s\neq i_1,\ldots,i_n}^N\lambda_s^2\prod_{j=1}^N\lambda_j^{2(n-1)}e^{nx}\Delta_{N-n}}
     {2^n\Delta_NE_{i_1}\cdots E_{i_n}},
       \end{split}
 \end{equation*}
which completes the proof by taking account of \eqref{Hankle determinant}.
\end{proof}

\begin{proposition}
The function $g$ satisfies
\begin{equation}\label{g-eqn}
    \sum_{n=0}^{N+1}\epsilon^{n(n-1)}h_ng^{N-n+1}=o(\epsilon^{N(N+1)}),
\end{equation}
where
 \[    h_0  =1, \qquad
       h_{N+1} =-\f{\prod_{i=1}^N\lambda_i^{2N}e^{Nx}}{2^N\Delta_N\prod_{i=1}^NE_i},\]
\begin{equation}\label{h_n def}
  h_n = d_ne^{(n-1)x}(e^xD_{N-n}^2-2D_{N-n+1}^0),\quad n=1,\ldots, N.
\end{equation}
\end{proposition}
\begin{proof}
By Lemma \ref{g_1,g_2},
  \begin{equation*}
    g^{N+1}+\sum_{n=1}^N\epsilon^{n(n-1)}g^{N-n+1}d_ne^{nx}D_{N-n}^2=g^N+
    2\sum_{n=1}^N\epsilon^{n(n+1)}g^{N-n}d_{n+1}e^{nx}D_{N-n}^0+o(\epsilon^{N(N+1)}).
  \end{equation*}
combining the equation in terms of $\epsilon^{n(n-1)}g^{N-n+1}\,(n=0,1,\ldots,N+1)$ leads to the conclusion.
\end{proof}

Taking derivatives of \eqref{g-eqn} with respect to $x$ and $t$, we have the following equation:
\begin{equation}\label{u asymptotic}
 -\frac{g_t}{g-g_x}=
  \frac{(\sum_{i=1}^Nc_i)g^{N+1}+\sum_{n=1}^N\epsilon^{(n-1)n}g^{N-n+1}d_ne^{(n-1)x}
  (e^xD_{N-n,t}^2-2D_{N-n+1,t}^0)+  o(\epsilon^{N(N+1)})}{g^{N+1}+\sum_{n=1}^N\epsilon^{(n-1)n}g^{N-n+1}d_ne^{nx}D_{N-n}^2
  +o(\epsilon^{N(N+1)})}.
\end{equation}
Next, taking account of the characteristics of \eqref{g-eqn}, to calculate $\lim_{\epsilon\rightarrow 0}\f{g_t}{g-g_x}$, we only need to seek positive series solution
\begin{equation}\label{g-series}
    g=g^{(0)}+g^{(2)}\epsilon^2+g^{(4)}\epsilon^4+\cdots, \quad g^{(m)}\geq0\,(m=0,2,\ldots).
\end{equation}

In the rest of this section, we will use $x_n$ to denote $\bar{x}_n$ given by \eqref{x_i} (see Remark \ref{x_n:explaination}).

\begin{proposition}\label{prop g leading term}
Suppose that $\{h_n\}$ satisfies
  \begin{equation}\label{h_n sign}
    h_1>0,\;h_2>0,\ldots, h_{n-1}>0,\;h_n\leq 0,\;h_{n+1}<0,\ldots ,h_N<0,
  \end{equation}
and the series \eqref{g-series} satisfies \eqref{g-eqn}, then
 \begin{equation}\label{g-leading term}
    g\sim g^{(2n-2)}\epsilon^{2n-2}=-\frac{h_n}{h_{n-1}}\epsilon^{2n-2},
 \end{equation}
and $x_{n-1}<x\leq x_n$\,$(n=1,2,\ldots,N)$.
\end{proposition}
The proof of Proposition of \ref{prop g leading term} need the following two lemmas.

\begin{lemma}\label{x_i sign}
For $h_n$ given by \eqref{h_n def}, we have
  \begin{equation}\label{x position}
    x=x_n \Leftrightarrow h_n=0,\quad\quad x>x_n \Leftrightarrow h_n>0,\quad \quad x<x_n \Leftrightarrow h_n<0.
  \end{equation}
\end{lemma}

\begin{lemma}\label{x_i order}
For $n=1,2,\ldots,N+1$,
  \begin{equation}\label{x_n-order}
    x_{n-1}<x_n,
  \end{equation}
with the convention that $x_0=-\infty,\;x_{N+1}=+\infty$.
\end{lemma}

Lemma \ref{x_i sign} follows from the definition of $h_n$. For Lemma \ref{x_i order}, we recall that $D_{N-n+1}^1>0$ by Lemma \ref{Jacobi formula}, and then \eqref{x_n-order} follows.
\begin{proof}[\bf{The proof of Proposition \ref{prop g leading term}}]
By Lemmas  \ref{x_i sign} and \ref{x_i order},  the condition \eqref{h_n sign} holds for some $n$, therefore, $x_{n-1}<x\leq x_n$. Substituting \[g=g^{(2m)}\epsilon^{2m}+O(\epsilon^{2m+2})\] into \eqref{g-eqn}, the lowest order of $\epsilon$ on left hand side of \eqref{g-eqn} is $m(2N-m+1)$. Note that
\[m(2N-m+1)\leq N(N+1),\quad 1\leq m\leq N,\]
 which leads to the coefficient of $\epsilon^{m(2N-m+1)}$ being zero, that is,
\[h_m(g^{(2m)})^{N-m+1}+h_{m+1}(g^{(2m)})^{N-m}=0.\]
 By \eqref{h_n sign} and $g^{(2m)}>0$, we have $m=n-1$ and \eqref{g-leading term}, which completes the proof.
\end{proof}

According to Lemmas \ref{x_i order} and \ref{x_i sign}, for any given $t$, if $x_{n-1}<x\leq x_n$\,$(n=1,2,\ldots,N)$ and the series \eqref{g-series} satisfies \eqref{g-eqn},
  then $g$ admits the estimate \eqref{g-leading term} on the order of $\epsilon$. Obviously,
    \[m(m-1)+(2n-2)(N-m+1)\geq (n-1)(2N-n+2),\quad 1\leq m\leq N,\]
    and the equality  holds for $m=n$ or $m=n-1$.
  Therefore, for any given $t$, the numerator and denominator of the right-hand side of \eqref{u asymptotic} share the same lowest power $\epsilon^{(n-1)(2N-n+2)}$ on the interval $x_{n-1}<x\leq x_n$, which come from the $(n-1)$-th and the $n$-th term, respectively, while they share the same lowest power $\epsilon^{N(N+1)}$ given by the $N$-th term on $x>x_N$.
Thus, we have the following proposition.

\begin{proposition}\label{u leading term}
When $\epsilon\rightarrow 0$ \,i.e. $\k\rightarrow 0$,
 \[\f{g_t}{g-g_x}=\f{d_Ne^{(N-1)x}(e^xD_{0,t}^2-2D_{1,t}^0)g^{(2N-2)}\epsilon^{N(N+1)}+o(\epsilon^{N(N+1)})}
    {-d_Ne^{Nx}D_0^2g^{(2N-2)}\epsilon^{N(N+1)}+o(\epsilon^{N(N+1)})},\qquad  x>x_N,\]
and
\[\f{g_t}{g-g_x}=-\f{G_n}{F_n},\qquad x_{n-1}<x\leq x_n\,(n=1,\ldots ,N),\]
with
\begin{equation*}
 \begin{split}
      &\begin{split}
    G_n=&\Big(d_ne^{(n-1)x}(e^xD_{N-n,t}^2-2D_{N-n+1,t}^0)(g^{(2n-2)})^{(N-n+1)}
     +d_{n-1}e^{(n-2)x}\\
    &(e^xD_{N-n+1,t}^2-2D_{N-n+2,t}^0)(g^{(2n-2)})^{(N-n+2)}\Big)\epsilon^{(n-1)(2N-n+2)}+o(\epsilon^{(n-1)(2N-n+2)}),
     \end{split}\\
    &\begin{split}
      F_n=&\Big(d_{n-1}e^{(n-1)x}D_{N-n+1}^2(g^{(2n-2)})^{(N-n+2)}+d_ne^{nx}
    D_{N-n}^2(g^{(2n-2)})^{(N-n+1)}
    \Big)\epsilon^{(n-1)(2N-n+2)}\\
   &+o(\epsilon^{(n-1)(2N-n+2)}).
     \end{split}
   \end{split}
\end{equation*}
\end{proposition}

With all the preparations above, we can complete the proof of theorems on convergence.

\begin{proof}[{\bf Proof of Theorem \ref{parameter limit}}]
According to Proposition \ref{u leading term}, $\lim_{\epsilon\rightarrow 0}\f{g_t}{g-g_x}$ exists. Using \eqref{g-leading term}, for $n=1,\ldots ,N$, we obtain
\begin{equation}\label{eqn:limit}
\lim_{\epsilon\rightarrow 0}\f{g_t}{g-g_x}
  =\f{e^{-x}d_{n-1}(e^xD_{N-n+1,t}^2-2D_{N-n+2,t}^2)h_n
-d_n(e^xD_{N-n,t}^2-2D_{N-n+1,t}^2)h_{n-1}}
{ -h_nd_{n-1}D_{N-n+1}^2+h_{n-1}d_ne^xD_{N-n}^2},
\end{equation}
on $x_{n-1}<x\leq x_n$.
Substituting $h_n$ defined by \eqref{h_n def} into \eqref{eqn:limit}, by Lemma \ref{critical identities} (replace $n$ in \eqref{Jacobi formula 1}-\eqref{Jacobi formula 4} with $N-n$), we have
\begin{equation*}
\lim_{\epsilon\rightarrow 0}\f{g_t}{g-g_x}=\f{
 e^xD_{N-n}^3+4e^{-x}D_{N-n+2}^{-1}}{D_{N-n+1}^1}, \quad x_{n-1}<x\leq x_n\,(n=1,\ldots ,N).
\end{equation*}
For $n=N+1$, taking account of $D_0^m=1$, $D_1^0=A_0$, we have
 \[\lim_{\epsilon\rightarrow 0}\f{g_t}{g-g_x}=2e^{-x}D_{1,t}^0=2\sum_{i=1}^Nc_iz_i^{-1},\quad
 x>x_N.\]

Now we conclude that  under appropriate scaling transformations,  $\f{g_t}{g-g_x}$ converges uniformly to $\bar{u}(x;t)$  when $\epsilon\rightarrow 0$ \,i.e. $\k\rightarrow 0$, which completes the proof.

\end{proof}

\begin{proof}[\bf{Proof of Theorem \ref{main result}}]

By Theorem \ref{parameter limit}, we need to show that
\begin{equation}\label{u_n}
  u_n= \sum_{i=1}^{n-1}m_ie^{-(x-x_i)}+\sum_{i=n}^Nm_ie^{-(x_i-x)},\quad n=1,\ldots ,N+1
\end{equation}
with $x_i,m_i$ given by \eqref{m_i}.

For $n=N+1$, by $D_{1,t}^0=2D_1^{-1}$, \eqref{u_n} is equivalent to
\[\sum_{i=0}^{N-1}\f{(D_{i+1}^0)^2}{D_{i+1}^1D_i^1=D_1^{-1}}.\]
For $n=1,\ldots,N$, \eqref{u_n} is equivalent to
\[\sum_{i=0}^{N-n}\f{(D_i^2)^2}{D_{i+1}^1D_i^1}=\f{D_{N-n}^3}{D_{N-n+1}^1},\qquad
\sum_{i=N-n+1}^{N-1}\f{(D_{i+1}^0)^2}{D_{i+1}^1D_i^1}=\f{D_{N-n+2}^{-1}}{D_{N-n+1}^1}.
 \]
The three identities follow readily from Corollary \ref{identities for main theorem} by replacing $n$ in \eqref{formula} with $-1$ and $N-n$, respectively.
\end{proof}

\section{Proof of Theorem \ref{new expression}}\label{proof2.1}

Consider the following determinant
\begin{equation}\label{A}
  A= \begin{vmatrix}
     e^{2\xi_1}+1&  e^{2\xi_2}-1 &  \cdots &e^{2\xi_N}+(-1)^{N-1} & (2\k)^N\\
     k_1e^{2\xi_1}-k_1&  k_2e^{2\xi_2}+k_2 &\cdots&k_Ne^{2\xi_N}+(-1)^{N-1}(-k_N)  & (2\k)^{N-1}\\
     k_1^2e^{2\xi_1}+k_1^2&  k_2^2e^{2\xi_2}-k_2^2 &\cdots&k_N^2e^{2\xi_N}+(-1)^{N-1}(-k_N)^2  & (2\k)^{N-2}\\
     \vdots & \vdots & \ddots &   \vdots &   \vdots \\
    k_1^N e^{2\xi_1}+(-k_1)^N& k_2^N e^{2\xi_2}-(-k_2)^N & \cdots &k_N^Ne^{2\xi_N}+(-1)^{N-1}(-k_N)^N &1\\
    \end{vmatrix}.
\end{equation}
 Easy to find that the determinant $A$ equals to the sum of all the terms $e^{2\xi_{i_1}+\cdots+2\xi_{i_n}}\,(n=1,2,\ldots,N)$ and the term without exponential functions.

\begin{lemma}\label{eqn:A}
For any $N\geq1$, the determinant $A$ can be expressed as follows:
  \begin{equation}\label{A-conjection}
    A=\prod_{i<j}^N(k_j-k_i)\prod_{i=1}^N(1+2\k k_i)+\sum_{n=1}^Ne^{2\xi_I}\Gamma_I\Gamma_J\prod_{i\in I,\, j\in J}(1-2\kappa k_i)(1+2\k k_j)(k_i+k_j).
\end{equation}
\end{lemma}

\begin{proof} We only need to show that the term without exponential functions equals to \begin{equation}\label{A-constant part}
\prod_{i<j}^N(k_j-k_i)\prod_{i=1}^N(1+2\k k_i),
\end{equation}
and for any $1\leq n\leq N$,\; I=$\{i_1<i_2\cdots<i_n\}$, the coefficient of the term $e^{2\xi_{i_1}+\cdots+2\xi_{i_n}}$ can be given by
\begin{equation}\label{A-exponent part}
\Gamma_I\Gamma_J\prod_{i\in I,\, j\in J}(1-2\k k_i)(1+2\k k_j)(k_i+k_j).
\end{equation}

 Let
\begin{equation*}
    \alpha_i=e^{2\xi_i}\left[
               \begin{array}{c}
                 1 \\
                 k_i \\
                 \vdots \\
                 k_i^N \\
               \end{array}
             \right],
             \qquad \beta_j=\left[
               \begin{array}{c}
                 (-1)^{j-1} \\
                (-1)^{j-1}(- k_j) \\
                 \vdots \\
                 (-1)^{j-1}(- k_j)^N\\
               \end{array}
             \right],\qquad\gamma=\left[
                                    \begin{array}{c}
                                      (2\k)^N \\
                                    (2\k)^{N-1} \\
                                      \vdots \\
                                      1 \\
                                    \end{array}
                                  \right]
\end{equation*}
then
\[A=|\alpha_1+\beta_1\quad \alpha_2+\beta_2\quad \cdots\quad\alpha_N+\beta_N\quad \gamma|,\]
and the term without exponential function equals to $|\beta_1\quad \beta_2\quad \cdots\quad\beta_N\quad \gamma|\triangleq \tilde{A}$.

Computing the determinant $\tilde{A}$, we have
\begin{equation}\label{A_0}
   \tilde{A}= (-1)^{\f{N(N-1)}{2}}\prod_{i<j}^N(k_i-k_j)\prod_{i=1}^N(1+2\k k_i)
       =\prod_{i<j}^N(k_j-k_i)\prod_{i=1}^N(1+2\k k_i),
\end{equation}
we have used the convention: $k_1<k_2<\cdots<k_N$ in the last step.

Replacing the columns of $\tilde{A}$ with indices $i_1,\ldots,i_n$ by $\alpha_{i_1},\ldots,\alpha_{i_n}$, respectively, we obtain a determinant B,
then the coefficient of the term $e^{2\xi_{i_1}+\cdots+2\xi_{i_n}}$ in determinant A can be given by B. Let $\beta_j^{\prime}=(-1)^{j-1}\beta_j\,(j\in J)$, $\gamma^{\;\prime}=(2\k)^{-N}\gamma$, replacing the columns of B with indices $j\,(j\in J),\,N+1$ by $\beta_j^{\prime}\,(j\in J)$ and $\gamma^{\;\prime}$, respectively, we obtain a Vandermonde determinant $C$. Besides, according to the operations above, we have
\begin{equation*}
    B=(-1)^{\f{N(N-1)}{2}-(i_1+\cdots+i_n-n)}(2\k)^NC,
\end{equation*}
with
\begin{equation*}
    \begin{split}
      C= & (-1)^{\f{(N-n)(N-n-1)}{2}}\Gamma_I\Gamma_J\prod_{l=1}^n\left(\prod_{j\in J,\,j<i_l}(k_{i_l}+k_j)
    \prod_{j\in J,\,j>i_l}(-k_j-k_{i_l})\right)\prod_{i\in I\,j\in J}(\f{1}{2\k}-k_i) (\f{1}{2\k}+k_j) \\
    =& (-1)^{\f{(N-n)(N-n-1)}{2}}\Gamma_I\Gamma_J (-1)^{nN-(i_1+\cdots+i_n)-\f{n(n-1)}{2}}\prod_{i\in I,\,j\in J}(k_i+k_j)(\f{1}{2\k}-k_i) (\f{1}{2\k}+k_j),
    \end{split}
\end{equation*}
 we have applied the convention: $i_1<i_2<\ldots<i_n$ and $k_1<k_2<\cdots<k_N$ multiple times to adjust the factors involving -1. Therefore
\[B=\Gamma_I\Gamma_J\prod_{i\in I,\, j\in J}(1-2\k k_i)(1+2\k k_j)(k_i+k_j),\]
which together with \eqref{A_0} completes  the proof of Lemma \ref{eqn:A}.
\end{proof}
Let $A^\prime$ be the determinant given by replacing the $N+1$ column of $A$ by
 \[\delta=\left(
           \begin{array}{cccc}
             1 & \f{1}{-2\k} & ( \f{1}{-2\k})^2 & ( \f{1}{-2\k})^N \\
           \end{array}
         \right)^T,
\]
Similar to the proof of  Lemma \ref{eqn:A}, we have
\begin{equation}\label{A prime-conjection}
 A^\prime=\prod_{i<j}^N(k_j-k_i)\prod_{i=1}^N(1-2\k k_i)+\sum_{n=1}^Ne^{2\xi_I}\Gamma_I\Gamma_J\prod_{i\in I,\, j\in J}(1+2\k k_i)(1-2\k k_j)(k_i+k_j).
\end{equation}
Furthermore, it is easy to check that \eqref{A-conjection} and \eqref{A prime-conjection} can be rewritten as
\begin{equation*}
A=\prod_{i<j}^N(k_j-k_i)\prod_{i=1}^N(1+2\k k_i)g_1,\quad\quad A^\prime=\prod_{i<j}^N(k_j-k_i)\prod_{i=1}^N(1-2\k k_i)g_2,
  \end{equation*}
with $g_1$  and $g_2$ given by \eqref{g_1} and \eqref{g_2}, respectively.
Taking account of the following identities
\begin{equation*} (2\k)^N2^N\exp\left(\sum_{i=1}^N\xi_i\right)W(\Phi_1,\ldots,\Phi_N,e^{\frac{y}{2\k}})
=e^{\frac{y}{2\k}}A,
 \end{equation*}
\begin{equation*}
(-2\k)^N2^N\exp\left(\sum_{i=1}^N\xi_i\right)W(\Phi_1,\ldots,\Phi_N,e^{-\frac{y}{2\k}})
=e^{-\frac{y}{2\k}}A^\prime,
  \end{equation*}
We obtain
\begin{equation}\label{relationship}
    (-1)^N\f{f_1}{f_2}=e^{\f{y}{\k}}\f{g_1}{g_2}\prod_{i=1}^Na_i.
\end{equation}
Note that $g_1>0,\,g_2>0$, hence,
\[\Big|\f{f_1}{f_2}\Big|=e^{\f{y}{\k}}\f{g_1}{g_2}\prod_{i=1}^Na_i,\]
For any given $N$,
 \begin{equation}\label{r}
    r=\f{\k f_1f_2}{\prod_{i=1}^N(k_i^2-\f{1}{4\k^2})}
\end{equation}
(see \cite{LZ, LYS})  is positive by the equation \eqref{relationship}. Thus, the reciprocal transformation:\,$dy=rdx-urdt$ has the inverse:\,$dx=\frac{1}{r}dy+udt$, and
\begin{equation}\label{reciprocal inverse}
    \f{\partial x}{\partial y}=\f{1}{r(y,t)},\qquad \f{\partial x}{\partial t}=u(y,t).
\end{equation}
is integrable, therefore,
\begin{equation*}
    \left(\f{1}{r}-\f{1}{\k}\right)_t=\left(\f{1}{r}\right)_t
    =u_y=\left(\ln\f{g_1}{g_2}\right)_{ty}.
\end{equation*}
Since $u,\,u_x,\,u_{xx}\rightarrow 0,\,r\rightarrow \k$ when $|x|\rightarrow\infty$, we have $u,\,u_y\rightarrow0,\,r\rightarrow \k$ when $|y|\rightarrow\infty$, which leads to
\begin{equation*}
    \left(\ln\f{g_1}{g_2}\right)_y+\beta(y)=\f{1}{r}-\frac{1}{\k},
\end{equation*}
where $\beta(y)$ is an arbitrary function of $y$, and $\beta(y)\rightarrow 0$ as $|y|\rightarrow\infty$, particularly, we can choose $\beta(y)=0$.
Thus, by the first equation in \eqref{reciprocal inverse}, we have \eqref{x(y,t)} with
 $\alpha$ an arbitrary integral constant.

We now complete the proof of Theorem \ref{new expression}.
\begin{remark}
In \cite {P2}, Parker provided the N-soliton solutions via Hirota bilinear method up to $N=4$ by means of Mathematica, these solutions take the same form as Theorem \ref{new expression}, but the proof for the general case $N>4$ has not been given.
\end{remark}

\section{Concluding remarks}
In this paper, we provide a new representation for the $N$-soliton solutions of the dispersion CH equation, obtained via Darboux transformation by Li. The new representation also provides a rigorous proof for bilinear solution of dispersion CH equation (see Theorems \ref{new expression} and \ref{equivalence}), which completes the work in \cite{P1, P2}. We also show that the $N$-peakon of the dispersionless CH equation can be recovered from a sequence of smooth $N$-soliton solutions of \eqref{kCH}, which establishes the relation between $N$-soliton solutions of \eqref{kCH} and $N$-peakon solutions of \eqref{CH eqn} for any $N\geq 1$. We believe that there is certain physics behind the mathematics we proved, while we still cannot explain it completely.

 In \cite{XZQ}, the authors also obtained smooth soliton solutions of \eqref{kCH} via Darboux transformation, while for the occurrence of the first and second derivatives of Wronskian, the multi-soliton they construct is very complicate, and it is difficult to investigate the convergence. Whether the multi-soliton solutions given by them are equivalent to the one studied in this paper is still unclear. If so, we may obtain the transformation between peakon solution of \eqref{CH eqn}, which is attractive.

 At present, the existence of soliton solutions with both positive and negative asymptotic speeds is still not clear, though the peakon-antipeakon (amplitudes $m_i$  in \eqref{N-peakon} have different signs) solutions of the dispersionless CH equation \eqref{CH eqn} have been given by the same formulas \eqref{N-peakon} in Lemma \ref{N-peakon}.  In \cite{LYS1}, the authors investigated multi-soliton solutions of dispersive CH equation, the formulas in \cite{LZ, LYS} may not give smooth soliton solutions,  the existence of soliton solutions with both positive and negative asymptotic speeds and the relation to the peakon-antipeakon solutions both are interesting topics for further study.


\begin{thebibliography}{99}

\bibitem {CH1993} Camassa R, Holm D D. An integrable shallow water equation with peaked solitons. Phys Rev Lett, 1993, 71: 1661--1664.

\bibitem {CHH1994} Camassa R, Holm D D, Hyman J M. A new integrable shallow water equation. Adv Appl Mech, 1994, 1--33.

\bibitem{Johnson2002}
Johnson R S. Camassa-Holm, Korteweg-de Vries and related models for water waves. J. Fluid Mech,  2002, 457: 63--82.

\bibitem {CL2009}
Constantin A, Lannes D. The hydrodynamical relevance of the Camassa-Holm and Degasperis-Procesi equations. Arch Ration Mech Anal, 2009, 192: 165--186.

\bibitem{C2000}
Constantin A. On the Scattering Problem for the Camassa-Holm Equation. Proc R Soc Lond A,  2001, 457: 953--970.

\bibitem {C1998isp}
Constantin A. On the inverse spectral problem for the Camassa-Holm equation. J Funct Anal, 1998, 155: 352--363.

\bibitem {RSJ}  Johnson R S. On solutions of the Camassa-Holm equation. Proc R Soc Lond A, 2003, 459: 1687--1708.

\bibitem {LZ} Li Y, Zhang J E. The multiple-soliton solution of the Camassa-Holm equation. Proc R Soc Lond A, 2004, 260: 2617--2627.

\bibitem {LYS}Li Y. Some water wave equations and integrability.  J Nonlinear Math Phys, 2005, (sup1)12; 466--481.

\bibitem {MY1} Matsuno  Y. Parametric representation for the multisoliton solution of the Camassa-Holm equation. J Phys Soc Japan, 2005, 74: 1983--1987.


\bibitem {P1} Parker A.  A factorization procedure for solving the Camassa-Holm equation.
 Inverse problems, 2006, 22: 599--609.

\bibitem {P2} Parker A. On the Camassa-Holm equation and a direct method of solution. III. N-soliton solutions.  Proc R Soc A, 2005, 461:3893--3911.

\bibitem {P3}Parker A. On the Camassa-Holm equation and a direct method of solution. II. Soliton solutions. Proc R Soc A,  2005, 461: 3611--3632.

\bibitem {P4}Parker A. On the Camassa-Holm equation and a direct method of solution. I. Bilinear form and solitary waves.  Proc R Soc Lond A, 2004, 460: 2929--2957.

\bibitem {XZQ} Xia B,  Zhou R, Qiao Z. Darboux transformation and multi-soliton solutions of the Camassa-Holm equation and modified Camassa-Holm equation. J Math Phys, 2016, 57, 103502.

\bibitem {Constantin2006}
 Constantin A.  The trajectories of particles in Stokes waves. Invent Math, 2006, 166: 523--535.

\bibitem {Constantin2007}
 Constantin A,  Escher J. Particle trajectories in solitary water waves. Bull  Amer  Math Soc, 2007, 44: 423--431.

\bibitem {Constantin2012}
  Constantin A. Particle trajectories in extreme Stokes waves. IMA J Appl Math, 2012,
   77: 293--307.

\bibitem {BSS2000}
Beals R, Sattinger D H,  Szmigielski J. Multipeakons and the classical moment problem. Adv Math, 2000, 154: 229--257.

\bibitem{C1998wavebreaking}
Constantin A,  Escher J. Wave breaking for nonlinear nonlocal shallow water equations. Acta Math, 1998, 181(2): 229--243.

\bibitem {BSS1999}
Beals R, Sattinger D H, Szmigielski J. Multipeakons and a theorem of Stieltjes. Inverse Problems, 1999, 15: 1--4.


\bibitem{C1998blowup}
Constantin A,  Escher J. Global existence and blow-up for a shallow water equation. Annali Della Scuola Normale Superiore Di Pisa Classe Di Scienze, 1998, 26: 303--328.


\bibitem {LiO1}
 Li Y A, Olver  P J. Convergence of solitary-wave solutions in a perturbed bi-Hamiltonian dynamical system: I. Compactons and peakons. Discrete  Contin Dyn Syst, 1997,
 3: 419--432.

\bibitem {LiO2}
 Li Y A, Olver  P J. Convergence of solitary-wave solutions in a perturbed bi-Hamiltonian dynamical system II: Complex analytic behavior and convergence to non-analytic solutions. Discrete  Contin Dyn Syst, 1998,  4: 159--191.

\bibitem {MY2}
Matsuno Y. The peakon limit of the N-soliton solution of the Camassa-Holm equation. J Phys Soc Japan, 2007, 76, 034003.


\bibitem {LYS1}
Li X,  Xu Y, Li Y. Investigation of multi-soliton, multi-cuspon solutions to the Camassa-Holm equation and their interaction.  Chin Ann  Math, 2012, 33B: 225--246.


\bibitem {VD} Vein  R,  Dale P. Determinants and their applications in mathematical physics,  Springer-Verlag, New York, 2006.


\end{thebibliography}

\end{document}